\documentclass[a4paper,11pt]{article}
\usepackage{amsmath}
\usepackage{amsfonts}
\usepackage{eucal}
\usepackage{amsthm}
\usepackage{amssymb,amscd}
\numberwithin{equation}{section}

\newcommand{\biggpare}[1]{\biggl(#1\biggr)}

\newcommand{\norm}[1]{\| #1 \|}

\newcommand{\abs}[1]{| #1 |}
\newcommand{\bigabs}[1]{\bigl| #1 \bigr|}

\newcommand{\jap}[1]{\langle #1 \rangle}

\def\a{\alpha}
\def\b{\beta}
\def\c{\gamma}
\def\d{\delta}
\def\e{\varepsilon}
\def\f{\varphi}

\def\l{\lambda}
\def\m{\mu}
\def\n{\nu}
\def\o{\omega}
\def\r{\rho}
\def\s{\sigma}
\def\t{\tau}

\def\y{\eta}
\def\z{\zeta}
\def\th{\theta}

\def\re{\mathbb{R}}
\def\co{\mathbb{C}}

\def\na{\mathbb{N}}
\def\pa{\partial}

\renewcommand{\Re}{\text{{\rm Re}\;}}
\renewcommand{\Im}{\text{{\rm Im}\;}}
\newcommand{\dom}{\mathcal{D}}

\DeclareMathOperator*{\slim}{s-lim}

\newcommand{\Ran}{\text{\rm Ran\;}}

\newtheorem{thm}{Theorem}
\newtheorem{lem}[thm]{Lemma}
\newtheorem{prop}[thm]{Proposition}

\theoremstyle{definition}

\newtheorem{ass}{Assumption}

\theoremstyle{remark}

\newtheorem*{thm*}{Theorem}
\newtheorem*{preack}{Acknowledgment}
\newenvironment{ack}{\begin{preack}\normalfont}{\end{preack}}

\theoremstyle{remark}
\newtheorem*{rem*}{Remark}

	\title{Geometric Scattering for Schr\"odinger Operators with Asymptotically Homogeneous Potentials of Order Zero}
	\author{
		Keita M{\sc ikami}%
		\footnote{Graduate School of Mathematical Sciences, 
			University of Tokyo, 3-8-1 Komaba, Meguro Tokyo, 
			153-8914 Japan. 
			E-mail: {\tt kmikami@ms.u-tokyo.ac.jp}.  } }
	
\begin{document}
		\maketitle
		\begin{abstract}
			In this paper we consider Schr\"odinger operators with potentials of order zero on asymptotically conic manifolds. We prove the existence and the completeness of the wave operators with a naturally defined free Hamiltonian.
		\end{abstract}
		
		\section{Introduction}
		Let $M$ be a smooth $n$-dimensional manifold. We suppose that $M$ is a union of a relatively compact part $M_{c}$ and a non-compact part $M_{\infty}$. We assume $M_{\infty}$ is diffeomorphic to $\re_+\times \pa M$ where $\pa M$ is a smooth compact manifold. By fixing an identification map, we identify a point $x$ in $M_{\infty}$ with $(r,\th)$ in $\re_+\times \pa M$. We also assume  $M_c\cap M_\infty$ is contained in $(0,1)\times \pa M$ under the above identification. We define a reference manifold by $M_f =\re \times \pa M$.
		
		We fix a local coordinate system $(U_\a , \f_\a , (\th_j)_{j=1}^{n-1})$ for $\pa M$ and define a local coordinate system for $M_{\infty}$ by $(\re_+\times U_\a , I\otimes \f_\a , (r,(\th_j)_{j=1}^{n-1}))$.
		
		Let $H(\th)$ be a positive smooth density on $\pa M$, and $G(x)$ be a positive smooth density on $M$ such that  $G(x)=r^{n-1}H(\th)$ in $\tilde M_{\infty}=(1,\infty) \times \pa M$. 
		We set $\mathcal{H}$ by $\mathcal{H}=L^2(M,G(x)\mathrm{d}x)$ and $\mathcal{H}_f$ by $\mathcal{H}_f=L^2(M_f,H(\th)\mathrm{d}r\mathrm{d}\th)$.

		We fix a smooth cut-off function $j \in C^\infty(\re)$ such that $j(r) = 1$ if $1 \leq \bigabs r$ and $j(r) = 0$ if $\bigabs r \leq \frac12$. We define $J: \mathcal{H}_f\to	\mathcal{H}$ by 
		\[
		(J\f)(x) =\begin{cases}
		 r^{-(n-1)/2}\, j(r) \, \f(r,\th) & \text{if }x=(r,\th)\in M_\infty\\
		 0 & \text{if }x\notin M_\infty,
		\end{cases}
		\]
		 for $\f\in \mathcal{H}_f$. 
		
		Let $P$ be an elliptic second order differential operator on $M$. We assume $P$ is bounded from below and $P$ is symmetric with $\dom(P)=C^\infty_c(M)$. We also assume $P$ is of the form,
		
		\[P= -\frac12 G^{-1} (\pa_r, \pa_\th/r) G \begin{pmatrix} a_1 & a_2 \\ {}^t a_2 & a_3 \end{pmatrix}
		\begin{pmatrix} \pa_r \\ \pa_\th/r \end{pmatrix} +V
		\quad \text{on } \tilde M_\infty,
		\]
		where $\{a_{i}\}^3_{i=1}$ are real-valued and positive definite smooth tensors on $M$, and $V$ is a smooth function on $M$. We assume $\{a_{i}\}^3_{i=1}$ and $V$ satisfy Assumption A.
		
		\begin{ass}
			Let $\m_1,\m_2,\m_3>1$
			\begin{itemize}
				\item{For any $\ell\in \na$, $\a\in\na^{n-1}$, there exists $C_{\ell \a}>0$ such that 
					\begin{align*}
					&\bigabs{\pa_r^\ell\pa_\th^\a (a_1(r,\th)-1)} \leq C_{\ell\a} r^{-\m_1-\ell}, \\
					&\bigabs{\pa_r^\ell\pa_\th^\a a_2(r,\th)} \leq C_{\ell\a} r^{-\m_2-\ell}, \\
					&\bigabs{\pa_r^\ell\pa_\th^\a (a_3(r,\th)-h(\th))} \leq C_{\ell\a} r^{-\m_3-\ell}
					\end{align*}
					on $\tilde M_\infty$ where $h(\th)=\{h^{j,k}(\th) \}$ is a positive symmetric smooth (2,0)-tensor on $\pa M$.}
				\item{
					$V$ is real-valued and has decomposition 
					\begin{align*}
					V(r,\th)=\tilde{V}(\th)+V_s(r,\th) \quad \text{if }(r,\th)\in \tilde M_\infty,
					\end{align*}
					where $V_s\in C^{\infty}(M)$ is real-valued and short-range, i.e., there exists $\m_4>1$ such that for any $\ell\in \na$, $\a\in\na^{n-1}$,
					\begin{align*}
					\bigabs{\pa_r^\ell\pa_\th^\a V_s(r,\th)} \leq C_{\ell\a} r^{-\m_4-\ell} \quad \text{if }(r,\th)\in \tilde M_\infty
					\end{align*}
					with a constant $C_{\ell \a}>0$. We also assume that $\tilde V\in C^{\infty}(\pa M)$ and that the set of critical value $\mathrm{Cv}(\tilde V)= \{\l\in\re \mid \l=\tilde V(\th), \pa_{\th} \tilde V(\th)=0\}$ is finite.
				}\end{itemize}
			\end{ass}
			\begin{rem*}
				Assumption A is a slightly more strict assumption for the perturbation of $a_2$ and $a_3$ than that of \cite{in}. This is necessary since we employ the smooth perturbation theory to prove the completeness of the wave operators.
			\end{rem*}
			
			We fix a smooth function $\tilde j \in C^\infty(\re)$ such that $\tilde j$ is strictly positive and $\tilde j(r) = 1$ if $\bigabs r \leq \frac12$ and $\tilde j(r) = \frac1{r^2}$ if $\bigabs r \geq 1$. We define the reference operator $P_f$ by
			\begin{align*}
			& P_f = -\frac12\frac{\pa^2}{\pa r^2} -\frac{\tilde j(r)}2 \sum_{j,k} H(\th)^{-1}\pa_{\th_j} H(\th) h^{jk}(\th) \pa_{\th_k} 
			+\tilde V(\th) \quad\text{on } M_f,
			\end{align*}
			We note that $P_f$ is elliptic from the assumption for $\tilde j$.
			
			\begin{thm}\label{thm:1}
				(1) The wave operators $W_\pm=\slim_{t\to\pm \infty} e^{itP} J e^{-itP_f}P_{ac}(P_f)$ exist,
				where $P_{ac}(\cdot)$ is the orthogonal projection onto the absolutely continuous subspace.\newline
				(2) The wave operators $W_\pm$ are complete, i.e., $\Ran W_\pm= \mathcal{H}_{ac}(P)$, where $\mathcal{H}_{ac}(\cdot)$ denotes the absolutely continuous subspace.  
			\end{thm}
			
			The study of the Schr\"odinger operators with potentials of order 0 was initiated by Herbst in \cite{h}, who characterized unitary equivalence of such operators and studied asymptotic behavior of time evolution in the Euclidean space case. He also investigated classical mechanics of Hamiltonian flow with potentials of order 0. In \cite{hs2}, Herbst and Skibsted compared Schr\"odinger operators with potentials of order 0 with Laplacian and proved the existence and asymptotic completeness of wave operators in high and low energy. 
			
			Agmon, Cruz, and Herbst considered a class of Schr\"odinger operators including Schr\"odinger operators with potentials of order 0. They showed Schr\"odinger operators with potentials of order 0 can be diagonalized in high energy by solving the eikonal equation in \cite{ach}(see also Sait$\mathrm{\bar{o}}$ \cite{ys}). They also showed that Schr\"odinger operators with potentials of order 0 satisfy the Mourre estimate with the modified conjugate operators.	
			
			One interesting property of the Schr\"odinger operators with potentials of order 0 is so called ``localization of the solution in direction'' which is proved by Herbst and Skibsted in \cite{hs2}. Let $P$ be a Schr\"odinger operator with potentials of order 0 on $\re^n$ and the variable $\th$ be in $S^{n-1}$. We define $\mathcal{H}_{\th} \subset L^2(\re^n)$ by $\{\f\in L^2(\re^n) \mid (\frac{x}{\abs{x}}-\th)e^{-itP}\f\rightarrow 0 \text{ as t} \rightarrow \infty\}$. Then it is shown that there exists $\{\th_m\}^{M}_{m=1}$ such that $L^2(\re^n) = \oplus^M_{m=1} \mathcal{H}_{\th_m}$.
			
			On the other hand, scattering theory for the Schr\"odinger operators with potentials of order 0 on the asymptotically conic manifolds is studied by Hassel, Melrose and Vasy in \cite{hmv} and \cite{hmv2} under the setting of \cite{mel}. They showed the asymptotic completeness and the localization of the solution in direction. 
			
			There are several formulation of geometric scattering and we employ the formulation of \cite{in}. A key idea of the formulation in \cite{in} is that they compare Schr\"odinger operators on the asymptotically conic manifolds with the simpler Schr\"odinger operators on the asymptotically tubic manifolds. In \cite{in2}, it is proved that scattering matrix of Schr\"odinger operators with short-range potential can be written as a Fourier integral operator associated with the asymptotic classical flow in the phase space.
			
			Scattering theory we propose is not a generalization of the results for the Euclidean case since we compare asymptotically the operators with the same leading asymptotic terms on different manifolds. We note that spectral properties presented in Section 2 are straight forward generalizations of previous results in \cite{ach} and \cite{h}. Our model is different from that of \cite{hmv} and \cite{hmv2}. They only treat the case when the potential is the Morse function and	they employ more strict assumption for the perturbation of $P$.
			
			\begin{ack}
				The author is grateful to his supervisor Professor Shu Nakamura for lots of discussions and parseverant supports and encouragement for this paper.	The author is under the support of the FMSP(Frontiers of Mathematics Science and Physics) program at the Graduate School of Mathematical Sciences, The University of Tokyo.
			\end{ack}
			
			\section{Spectral properties of $P$ and $P_f$}
			
			\subsection{Preparation for the proof of the Mourre estimate.}
			From Corollary 4 of \cite{in}, we learn $P$ is self-adjoint. We can prove that $P_f$ is self-adjoint similarly to the proof of Proposition 3 in \cite{in} since $P_f$ is elliptic.
			
			We define a smooth vector field $X_\l$ on $M$ for $\l>0$ by 
			\begin{align*}
			& X_\l=
			\begin{cases}
			 j({r})\, r\, \frac{\pa}{\pa r} -\sum_{j,k} \frac1\l j(r)(\pa_{\th_j}\tilde V(\th) ) h^{jk}(\th)\pa _{\th_k}  & \text{ on } M_\infty\\
			 0 & \text{ on } M_c.
			\end{cases}			
			\end{align*}
			Then $X_\l$ generates flow $\mathrm{exp}(-tX_\l)$ for $t \in \re$. We define the unitary group $U_\l(t)$ on $\mathcal{H}$
			as follows:
			\begin{align*}
			& U_\l(t)\f(x) = \Phi(t,x) \f(\mathrm{exp}(-tX_\l)x),
			\end{align*}
			where $\Phi(t,x)$ is a positive and smooth weight function to make $U_\l(t)$ unitary.
			Since $U_\l(t)$ is unitary, it can be written as $U_\l(t)=e^{itA_\l}$ where $A_\l$ is a self-adjoint operator. From direct calculation, we can write $A_\l$ as follows on $C^\infty_c(M)$:
			\begin{align*}
			& A_\l=\frac{1}{2i} G(r,\th)^{-\frac12}\biggpare{j({r})\, r\, \frac{\pa}{\pa r}+ \frac{\pa}{\pa r} j(r)\, r-\sum_{j,k} \frac1\l j(r)(\pa_{\th_j}\tilde V(\th) )h^{jk}(\th)\pa _{\th_k}\\
				& \quad -\sum_{j,k}\frac1\l \pa_{\th_j}h^{jk}(\th) (\pa _{\th_k}\tilde V(\th))j(r)}G(r,\th)^{\frac12}\\
			& = \frac{1}{2i} \biggpare{j({r})\, r\, \frac{\pa}{\pa r}+ \frac{\pa}{\pa r} j(r)\, r +(n-1)j(r)- \sum_{j,k} \frac1\l j(r)(\pa_{\th_j}\tilde V(\th) )h^{jk}(\th)\pa _{\th_k}\\
				& \quad -\sum_{j,k}\frac1\l \pa_{\th_j}h^{jk}(\th) (\pa _{\th_k}\tilde V(\th))j(r) -\sum_{j,k}\frac1\l (\pa_{\th_j}H(\th))h^{jk}(\th) (\pa _{\th_k}\tilde V(\th))j(r)}.
			\end{align*}
			By the similar argument, we define a self-adjoint operator $A_f$ on $\mathcal{H}_f$ as follows:
			
			\begin{align*}
			& A_f=\frac{1}{2i} H(\th)^{-\frac12}\biggpare{j({r})\, r\, \frac{\pa}{\pa r}+ \frac{\pa}{\pa r} j(r)\, r-\sum_{j,k} \frac1\l j(r)(\pa_{\th_j}\tilde V(\th) )h^{jk}(\th)\pa _{\th_k}\\
				& \quad -\sum_{j,k}\frac1\l \pa_{\th_j}h^{jk}(\th) (\pa _{\th_k}\tilde V(\th))j(r)}H(\th)^{\frac12}\\
			& = \frac{1}{2i} \biggpare{j({r})\, r\, \frac{\pa}{\pa r}+ \frac{\pa}{\pa r} j(r)\, r - \sum_{j,k} \frac1\l j(r)(\pa_{\th_j}\tilde V(\th) )h^{jk}(\th)\pa _{\th_k}\\
				& \quad -\sum_{j,k}\frac1\l \pa_{\th_j}h^{jk}(\th) (\pa _{\th_k}\tilde V(\th))j(r) -\sum_{j,k}\frac1\l (\pa_{\th_j}H(\th))h^{jk}(\th) (\pa _{\th_k}\tilde V(\th))j(r)}.
			\end{align*}
			
			The main part of this section is to prove the following theorem.
			\begin{thm}\label{thm:2}
				Let $P$ satisfies Assumption A and we assume $\l$ is sufficiently large. 
				Then \newline
				(1) The point spectrum of $P$ is discrete in $\re$. \newline
				(2) Let $I\Subset \re\setminus (\s_{pp}(P) \bigcup \mathrm{Cv}(\tilde V))$ and let $s>1/2$. Then 
				\[
				\sup_{\Re \! z\in I, \Im \! z\neq 0} \norm{\jap{A_\l}^{-s} (P-z)^{-1} \jap{A_\l}^{-s}} <\infty.
				\]
				In particular, $\s_{sc}(P)=\emptyset$ and $\jap{A_\l}^{-s} (P-z)^{-1} \jap{A_\l}^{-s}$ is H\"older continuous in $z$ if $\Re \! z\in I$ and $\Im \! z\neq 0$. \newline
				(3) Let $I$ be as above, and let $s>1/2$. Then 
				\[
				\sup_{\Re \! z\in I, \Im \! z\neq 0} \norm{\jap{r}^{-s} (P-z)^{-1} \jap{r}^{-s}} <\infty.
				\]
				In particular, $\jap{r}^{-s} (P-z)^{-1} \jap{r}^{-s}$ is H\"older continuous in $z$ if $\Re \! z\in I$ and $\Im \! z\neq 0$. Thus
				\[
				\jap{r}^{-s} (P-E\pm i0)^{-1}\jap{r}^{-s} 
				=\lim_{\e\to+0}\jap{r}^{-s} (P-E\pm i\e)^{-1}\jap{r}^{-s}
				\]
				converge uniformly in $E\in I$. 
			\end{thm}
			
			We can prove the same theorem for $P_f$, i.e.,
			
			\begin{prop}
				Theorem 2 holds for $P_f$ with $A_\l$ replaced by $A_f$.
			\end{prop}
			
			From the Mourre theorem in \cite{ya}, it suffices to prove following Theorem.
			\begin{thm}\label{prop:1}
				Let $P$ and $A_\l$ be as above and we assume $\l$ is sufficiently large. For each interval $I\Subset \re\setminus (\s_{pp}(P) \bigcup \mathrm{Cv}(\tilde V))$, there exist $\a >0$ and a compact operator $K$ such that
				\[
				\chi_I(P) i[P,A_\l]\chi_I(P) \geq \a \chi_I(P) + K, 
				\]
				where $\chi_I$ is the indicator function of $I$. 
			\end{thm}
			
			Let $m$ be real number. We consider a symbol class $S^m(T^*M_\infty)$ such that $a \in S^m(T^*M_\infty)$ means $a \in C^{\infty}(T^*M_\infty) = C^{\infty}(T^*\re_+ \times T^*\pa M)$ and for each multi-index $\a,\b,\c,\d$ there exists $C_{\a,\b,\c,\d}>0$ with
			\begin{align*}
			& \bigabs{\pa_r^\a\pa_\rho^\b\pa_\th^\c\pa_\o^\d a(r,\rho,\th,\o)} \leq 
			C_{\a,\b,\c,\d} \jap{r}^{m-|\a|}\jap{\rho}^{m-|\c|}\jap{\o}^{m-|\d|} .
			\end{align*}
			
			We quantize $a \in S^m(T^*M_\infty)$ by following \cite{in2}. Let $\{\f_\a, U_\a\}$ be a local coordinate system of $\pa M$ and $\{\chi^2_\a\}$ be a partition of unity on $\pa M$ compatible with our coordinate system, i.e., $\chi_\a\in C_0^\infty(U_\a)$ and $\sum_{\a} \chi_\a(\th)^2 \equiv 1$ on $\pa M$. We also denote $\tilde \chi_\a(r,\th)= \chi_\a(\th) j(r)\in C^\infty(M_\infty)$. Let	$u\in C_0^\infty(T^*M)$. We denote by $a_{(\a)}$ and $G_{(\a)}$ the representation of $a$ and $G$ in the local coordinate 
			$(1\otimes \f_\a,\re\times U_\a)$, respectively. 
			We quantize $a$ by 
			\begin{align*}
			Op^W(a)u = \sum_\a \tilde \chi_\a G_{(\a)}^{-1/2} a^W_{(\a)}(r, D_r, \th, D_\th) G_{(\a)}^{1/2}
			\tilde \chi_\a u
			\end{align*}
			
			From general theory of the pseudodifferential operators in \cite{zws}, we directly obtain following Proposition:
			\begin{prop}
				Let $a \in S^m(T^*M_\infty)$, then $Op^W(a)$ is a compact operator on $\mathcal{H}$ if $m<0$.
			\end{prop}
			
			From Weyl calculus in \cite{in2} and \cite{zws}, we can represent $A_\l$ and $P$ as a pseudodifferential operator.
			\begin{prop}\label{prop:a}
			 Let 
			 \begin{align*}
			 &p(r,\rho,\th,\o)=\frac12 (\rho, \o/r) \begin{pmatrix} a_1(r,\th) & a_2(r,\th) \\ {}^t a_2(r,\th) & a_3(r,\th) \end{pmatrix}
			 \begin{pmatrix} \rho \\ \o/r \end{pmatrix} +V(r,\th) \text{ and } \\
			 &a_\l(r,\th)=j({r})\, r\, \rho -\sum_{j,k} \frac1\l j(r)(\pa_{\th_j}\tilde V(\th) ) h^{jk}(\th)\o_k.
			 \end{align*}
			 Then we obtain following formula:
			 \begin{align*}
			 &Op^W(p) = P + E_1 \text{ and } Op^W(a_\l)=A_\l+E_2,
			 \end{align*}
			 where $E_1,E_2$ satisfy the condition $\chi_I(P) E_i \chi_I(P)$ are compact for $i=1,2$.
			\end{prop}
			
			\subsection{Proof of the Mourre estimate.}
			The proof of Theorem \ref{prop:1} is based on that of Euclidean case in Appendix C of \cite{ach} and we only prove Theorem 2 since the proof of Proposition 3 is exactly the same with that of Theorem 2.
			
			 We have to prepare lemma to prove Theorem 4. Let $P_0=P-j(2r)\tilde V(\th)$.
			
			\begin{lem}\label{lem:1}
				Let $P$, $\l$, and $I$ satisfy the assumption for Theorem 4.
				Then,
				\begin{align*}
				& i[P , A_\l] \geq P_0 + \frac1\l \sum_{j,k}  j(r)(\pa_{\th_j}\tilde V) h^{jk}(\th)(\pa_{\th_k}\tilde V) + E, 
				\end{align*}
				where $E$ satisfies the condition $\chi_I(P) E \chi_I(P)$ is compact.
			\end{lem}
			\begin{proof}
				From, Proposition \ref{prop:a} and locally compactness of $P$, we obtain
				\begin{align*}
				& i[P , A_\l] =i[Op^W(p)-V, Op^W(a_\l)] + \frac{j(r)}\l \sum_{j,k}  (\pa_{\th_j}\tilde V) h^{jk}(\th)(\pa_{\th_k}\tilde V) + E_3, 
				\end{align*}
				where $E_3$ satisfies the condition $\chi_I(P) E_3 \chi_I(P)$ is compact.
				
				From the formula for the commutator of the pseudodifferential operators in \cite{zws}, we see
				\begin{align*}
				& [Op^W(p) , Op^W(a_\l)] =Op^W(\{p,a_\l\}) +E_4,
				\end{align*}
				where $E_4$ satisfies the condition $\chi_I(P) E_4 \chi_I(P)$ is compact and  $\{\cdot,\cdot\}$ denotes the Poisson bracket.
				
				Let $x=(r,\th)$ and $\z=(\rho,\o)$ , we obtain the followings:
				\begin{align*}
				& 2\pa_x (p - V) = (\rho, \o) \{ \pa_x \begin{pmatrix} a_1(r,\th) & a_2(r,\th)/r \\ {}^t a_2(r,\th)/r & a_3(r,\th)/{r^2} \end{pmatrix}\}
				\begin{pmatrix} \rho \\ \o \end{pmatrix} , \\
				& 2\pa_\z (p - V)= 2 \begin{pmatrix} a_1(r,\th) & a_2(r,\th) \\ {}^t a_2(r,\th) & a_3(r,\th)/r \end{pmatrix}
				\begin{pmatrix} \rho \\ \o/r \end{pmatrix} , \\
				& 2\pa_x a_\l= \Bigl( \{\pa_r (j({r}) r) \}\rho -\sum_{j,k} \frac{(\pa_r j(r))}\l (\pa_{\th_j}\tilde V(\th) ) h^{jk}(\th)\o_k , \\
				& \qquad \qquad \qquad \qquad \qquad \frac{j(r)}\l \sum_{j,k}  {}^t\pa_\th\{(\pa_{\th_j}\tilde V(\th) ) h^{jk}(\th)\}\o_k \Bigr) , \\
				& 2\pa_\z a_\l= \Bigl(j({r})r,(-\sum_{j} \frac1\l j(r)(\pa_{\th_j}\tilde V(\th) ) h^{jk}(\th))_k \Bigr) , \\
				& \{p-V,a_\l\}=\pa_x a_\l\pa_\z (p - V) - \pa_\z a_\l\pa_x (p - V) \\
				& = (\rho, \o/r) \begin{pmatrix} a_1(r,\th) & a_2(r,\th) \\ {}^t a_2(r,\th) & a_3(r,\th) \end{pmatrix}
				\begin{pmatrix} \rho \\ \o/r \end{pmatrix} -\frac1\l (\rho, \o/r) B	\begin{pmatrix} \rho \\ \o/r \end{pmatrix} +E_5,
				\end{align*}
				where $B=\mathcal{O}(1)$ and $E_5$ satisfies the condition $\chi_I(P) E_5 \chi_I(P)$ is compact.
				By taking $\l$ sufficiently large, we obtain
				\begin{align*}
				& \{p-V,a_\l\} > P_0 +E,
				\end{align*}
				where $E$ satisfies the condition $\chi_I(P) E \chi_I(P)$ is compact. This concludes the assertion.
			\end{proof}
			
			\begin{proof}[Proof of Theorem \ref{prop:1}]
				Since I is an interval, we can take $\n\in\re$ and $\t\geq 0$ such that $I=(\n-\t,\n+\t)$. Let $\chi \in C^{\infty}_c(\re)$ be such that  $\chi(x)=1$ if $x\in I'=(\n-\t',\n+\t') \Subset I$ and $ \mathrm{supp}\chi \Subset I$. Then it suffices to prove Theorem 4 with $\chi_I$ is replaced by such $\chi$. For the technical convince, we assume $\chi$ satisfies this condition with $I'$ and $\t'$ are replaced by $I$ and $\t$.
				
				From Lemma 5, it is sufficient to prove that
				\[
				\chi(P) \{P_0 + \frac1\l \sum_{j,k}  j(r)(\pa_{\th_j}\tilde V) h^{jk}(\th)(\pa_{\th_k}\tilde V)\}\chi(P) \geq \a \chi(P) + K.
				\]
				
				Since $I\Subset \re\setminus (\s_{pp}(P) \bigcup \mathrm{Cv}(\tilde V))$, we can find $\y>0$ such that\\
				\[ 
				I'= (\n-\t-\y,\n+\t+\y) \Subset \re\setminus (\s_{pp}(P) \bigcup \mathrm{Cv}(\tilde V)).
				\]
				
				We define $\mathcal{O}_i \subset \pa M$ ($i=1,2,3$) by 
				\begin{align*}
				&\mathcal{O}_1=\tilde V^{-1}(\n-\t-\y,\n+\t+\y),\\
				&\mathcal{O}_2=\tilde V^{-1}(-\infty,\n-\t-\frac12\y),\\
				&\mathcal{O}_3=\tilde V^{-1}(\n+\t+\frac12\y,\infty)
				\end{align*}
				
				Since $\pa M$ is compact and smooth, we can take a partition of unity $\{\tilde f_i\}_{i=1,2,3}$ compatible with  $\{\mathcal{O}_i\}_{i=1,2,3}$ i.e. there exist $\tilde f_i\in C^\infty(\pa M)$ such that $\mathrm{supp}(\tilde f_i) \subset \mathcal{O}_i$ for $i=1,2,3$ and $\sum_{i=1}^{3} \tilde f^2_i =1$.
				We define $f_i \in C^\infty(M)$ by $f_i(r,\th) =j(r)\tilde f_i(\th)$ and $\psi \in C^\infty_c(M)$ by $\psi=1-\sum_{i=1}^{3} f^2_i$.
				
				Then we obtain
				\begin{align*}
				& \chi(P) \{P_0 + \frac1\l \sum_{j,k}  j(r)(\pa_{\th_j}\tilde V) h^{jk}(\th)(\pa_{\th_k}\tilde V)\}\chi(P)\\
				& = \chi(P) \sum_{i=1}^{3} f^2_i \{P_0 + \frac1\l \sum_{j,k}  j(r)(\pa_{\th_j}\tilde V) h^{jk}(\th)(\pa_{\th_k}\tilde V)\}\chi(P) +K_1\\
				& = \sum_{i=1}^{3} \chi(P)  f_i \{P_0 + \frac1\l \sum_{j,k}  j(r)(\pa_{\th_j}\tilde V) h^{jk}(\th)(\pa_{\th_k}\tilde V)\} f_i \chi(P) +K_2		
				\end{align*}
				where $K_1$ and $K_2$ are compact operators and we have used the Helffer-Sj\"ostrand formula from Section 14 of \cite{zws} in the third line.
				
				Since $\pa_\th \tilde V \neq 0$ on $\mathrm{supp} f_1$, there exists $\d>0$ such that
				\begin{align*}
				& \sum_{j,k} (\pa_{\th_j}\tilde V) h^{jk}(\th)(\pa_{\th_k}\tilde V) > \d
				\end{align*}
				Then we obtain,
				\begin{align*}
				& \chi(P)  f_1 \{P_0 + \frac1\l \sum_{j,k}  j(r)(\pa_{\th_j}\tilde V) h^{jk}(\th)(\pa_{\th_k}\tilde V)\} f_1 \chi(P)\\
				& \geq \frac\d\l \chi(P)  f^2_1\chi(P) + K_3
				\end{align*}
				where $K_3$ is a compact operator.
				
				Since $ \n-\tilde V \geq \t + \frac\y2$ on  $\mathrm{supp}f_2$ and $\chi(P)(P-\n)\chi(P) \geq \t\chi(P)^2 $, we obtain the following
				\begin{align*}
				& \chi(P)  f_2 \{P_0 + \frac1\l \sum_{j,k}  j(r)(\pa_{\th_j}\tilde V) h^{jk}(\th)(\pa_{\th_k}\tilde V)\} f_2 \chi(P)\\
				& \geq \chi(P)  f_2 P_0 f_2\chi(P) + K_4\\
				& \geq \chi(P)  f_2 (P - \n +\n -\tilde V) f_2\chi(P) + K_4\\
				& \geq \chi(P)  f_2 (P - \n + \t + \frac\y2) f_2\chi(P) + K_4\\
				& = f_2 \chi(P)   (P - \n + \t + \frac\y2) \chi(P) f_2 +K_5\\
				& \geq f_2 \chi(P)   (-\t + \t + \frac\y2) \chi(P) f_2 +K_5\\
				& = \frac\y2\chi(P)  f^2_2 \chi(P) +K_6
				\end{align*}
				where $K_4$, $K_5$ and $K_6$ are compact operators.
				
				Since $ \n-\tilde V \leq - \t - \frac\y2$ on  $\mathrm{supp}f_3$ and $\chi(P)(P-\n)\chi(P) \leq \t\chi(P)^2 $, we obtain the following
				\begin{align*}
				& \chi(P)  f_3 P_0 f_3 \chi(P)\\
				& = \chi(P)  f_3 (P - \n +\n -\tilde V) f_3\chi(P) \\
				& \leq \chi(P)  f_3 (P - \n - \t - \frac\y2) f_3\chi(P) \\
				& = f_3 \chi(P)   (P - \n - \t - \frac\y2) \chi(P) f_3 +K_7\\
				& \leq f_3 \chi(P)   (\t - \t - \frac\y2) \chi(P) f_3 +K_7\\
				& = -\frac\y2\chi(P)  f^2_3 \chi(P) +K_8,
				\end{align*}
				where $K_7$ and $K_8$ are compact operators.
				Thus we obtain, 
				\begin{align*}
				& 0 \leq \chi(P)  f_3 (P_0 + \frac\y2) f_3 \chi(P) \leq K_8
				\end{align*}
				and hence $\chi(P)  f_3 (P_0 + \frac\y2) f_3 \chi(P) $ is a compact operator.
				Let $\c<\frac\y2$, then we see
				\begin{align*}
				&  \c\chi(P)  f^2_3 \chi(P)-\chi(P)  f_3 (P_0 + \frac\y2) f_3\chi(P)\\
				& = -\chi(P)  f_3 P_0 f_3\chi(P)-(\frac\y2-\c)\chi(P)  f^2_3\chi(P)\\
				& \leq 0 \\
				& \leq \chi(P)  f_3 \{P_0 + \frac1\l \sum_{j,k}  j(r)(\pa_{\th_j}\tilde V) h^{jk}(\th)(\pa_{\th_k}\tilde V)\} f_3 \chi(P). 
				\end{align*}
				Since $\chi(P)  f_3 (P_0 + \frac\y2) f_3 \chi(P) $ is a compact operator, we get inequality for $f_3$.
				
				Finally, by taking $\a = \min\{\frac\d\l,\frac\y2, \c\}$, we conclude the assertion.
			\end{proof}
			
			\begin{proof}[Proof of Theorem \ref{thm:2}]
				From Theorem 4 and the Mourre theory in \cite{ya}, we only have to prove the following:
				\begin{enumerate}
					\item {$P$ is $C^2(A_\l)$ class i.e. $[[(P-z)^{-1},A_\l],A_\l]$ is bounded for any $z \in \r(P)$.}
					\item {$\jap{r}^{-s} (P-E\pm i0)^{-1}\jap{r}^{-s} 
						=\lim_{\e\to+0}\jap{r}^{-s} (P-E\pm i\e)^{-1}\jap{r}^{-s}$ converge uniformly in $E\in I$.}
				\end{enumerate}
				
				We see, from (2.10) in the proof of lemma 7, $[P,A_\l]$ is $P$-bounded. We can calculate $[[P,A_\l],A_\l]$ is $P$-bounded similarly. Since
				\begin{align*}
				& [(P-z)^{-1},A_\l],A_\l]\\
				& =(P-z)^{-1}[P,A_\l][(P-z)^{-1},A_\l] + (P-z)^{-1}[[P,A_\l],A_\l](P-z)^{-1}\\
				& + [(P-z)^{-1},A_\l] [P,A_\l](P-z)^{-1},
				\end{align*}
				we learn $[[(P-z)^{-1},A_\l],A_\l]$ is bounded.	 
				
				Concerning claim 2, we assume the following lemma for the moment.
				\begin{lem}\label{lem:2}
					$\jap {A_\l}^s (P\pm i)^{-1}\jap r^{-s}$ are bounded operators on $\mathcal{H}$.
				\end{lem}
				When Im$z\geq 1$, claim 2 is obvious and hence we may assume Im$z \leq 1$. Then we can calculate 
				\begin{align*}
				& (P-z)^{-1}\\
				& = (P+i)^{-1} +(z+i)(P+i)^{-2} +(z+i)^2(P+i)^{-1}(P-z)^{-1}(P+i)^{-1}
				\end{align*}
				Since first and second term in the right hand side are uniformly bounded if Im$z \leq 1$, we only have to treat the third term. Concerning the third term, we see
				\begin{align*}
				& \jap r^{-s}(z+i)^2(P+i)^{-1}(P-z)^{-1}(P+i)^{-1} \jap r^{-s}\\
				& =(z+i)^2\jap r^{-s}(P + i)^{-1}\jap {A_\l}^s\\
				& \quad \times+ \jap{A_\l}^{-s} (P-z)^{-1} \jap{A_\l}^{-s}\\
				& \quad \times  \jap {A_\l}^s (P + i)^{-1}\jap r^{-s} .
				\end{align*}
				From Theorem \ref{thm:2}, we learn $\jap{A_\l}^{-s} (P-z)^{-1} \jap{A_\l}^{-s}$ is uniformly bounded and thus right hand side of the above equality is uniformly bounded and this completes the proof of Theorem 2.
			\end{proof}
			
			\begin{proof}[Proof of Lemma \ref{lem:2}]
				From the usual interpolation argument, it suffices to prove Lemma \ref{lem:2} when $s=0,1$. When $s=0$, it is obvious. When $s=1$, we see
				\begin{align*}
				& A_\l (P\pm i)^{-1}\jap r^{-1} = (P\pm i)^{-1}[P,A_\l](P\pm i)^{-1}\jap r^{-1} +  (P\pm i)^{-1}A_\l r^{-1}.
				\end{align*}
				Proof of lemma 5 shows the right hand side of the above equality is bounded. Thus $A_\l (P\pm i)^{-1}\jap r^{-1}$ is also bounded, which concludes the proof.
			\end{proof}
			
		\subsection{Properties of $(PJ-JP_f)$.}
		We now prove some properties of $(PJ-JP_f)$.
		\begin{prop}\label{prop:2}
			(1) $PJ-JP_f$ can be written as following,
			\begin{align*}
				& (PJ-JP_f) = - (\tilde J\pa_r + \pa_r\tilde J)\f + \frac{(n-1)^2-2(n-1)}{4r^2}J\f + \frac{n-1}{2r} \tilde J\f\\
				& -\frac12 G^{-1} (\pa_r, \pa_\th/r) G \begin{pmatrix} a_1-1 & a_2 \\ {}^t a_2 & a_3-r^2\tilde j h \end{pmatrix}
				\begin{pmatrix} \pa_r \\ \pa_\th/r \end{pmatrix}J +V_sJ,
			\end{align*}
			
			where $\tilde J: \mathcal{H}_f\to
			\mathcal{H}$ is such that
			\[
			(\tilde J\f)(r,\th) = r^{-(n-1)/2}\,(\pa_r j(r)) \, \f(r,\th)
			\quad \text{if }(r,\th)\in M_\infty, 
			\]
			and
			\[\tilde  J\f(x)=0 \text{ if } x\notin M_\infty .
			\]
			(2) $(P-z)^{-1}\jap{r}^s(PJ-JP_f)\jap{r}^s$ is bounded for some $1/2<s<1$ and any $z\in\rho(P)$.\\
			(3) $\jap{r}^s\chi(P)(PJ-JP_f)\jap{r}^s$ is bounded for some $1/2<s<1$ and any $\chi\in C^\infty_c(\re)$.
		\end{prop}
		\begin{proof}
			From direct computations, for $\f \in C^\infty_c(M_f)$, we learn
			\begin{align*}
				& J \pa_r^2 \f\\
				& = r^{-(n-1)/2}\, j(r) \pa^2_r\f\\
				& = \pa_r r^{-(n-1)/2}\, j(r)\pa_r \f \\
				& + \{-r^{-(n-1)/2}(\pa_r  j(r)) + \frac{n-1}{2r}\times r^{-(n-1)/2} j(r)\}\pa_r\f\\
				& = \pa^2_r r^{-(n-1)/2}\, j(r)\f \\
				& + \{-r^{-(n-1)/2}(\pa_r  j(r)) + \frac{n-1}{2r} \times r^{-(n-1)/2}  j(r)\}\pa_r\f\\
				& + \pa_r\{-r^{-(n-1)/2}(\pa_r  j(r)) + \frac{n-1}{2r} \times r^{-(n-1)/2}  j(r)\}\f\\
				& = r^{-(n-1)}\pa_r r^{n-1}\pa_r r^{-(n-1)/2}\, j(r)\pa_r \f \\
				& - \{r^{-(n-1)/2}(\pa_r  j(r))\pa_r + \pa_rr^{-(n-1)/2}(\pa_r  j(r))  \}\f\\
				& + \frac{(n-1)^2-2(n-1)}{4r^2}\times r^{-(n-1)/2} j(r)\f\\
				& + \frac{n-1}{2r} \times r^{-(n-1)/2}(\pa_r  j(r))\f.
			\end{align*}
			Thus we obtain,
			\begin{align*}
				& (PJ-JP_f)\f\\
				& =- \{r^{-(n-1)/2}(\pa_r  j(r))\pa_r + \pa_rr^{-(n-1)/2}(\pa_r  j(r))  \}\f\\
				& + \frac{(n-1)^2-2(n-1)}{4r^2}\times r^{-(n-1)/2} j(r)\f\\
				& + \frac{n-1}{2r} \times r^{-(n-1)/2}(\pa_r  j(r))\f\\
				& -\frac12 G^{-1} (\pa_r, \pa_\th/r) G \begin{pmatrix} a_1-1 & a_2 \\ {}^t a_2 & a_3-r^2\tilde j h \end{pmatrix}
				\begin{pmatrix} \pa_r \\ \pa_\th/r \end{pmatrix}J +V_sJ\\
				&=- (\tilde J\pa_r + \pa_r\tilde J)\f + \frac{(n-1)^2-2(n-1)}{4r^2}J\f + \frac{n-1}{2r} \tilde J\f\\
				& -\frac12 G^{-1} (\pa_r, \pa_\th/r) G \begin{pmatrix} a_1-1 & a_2 \\ {}^t a_2 & a_3-r^2\tilde j h \end{pmatrix}
				\begin{pmatrix} \pa_r \\ \pa_\th/r \end{pmatrix}J +V_sJ,
			\end{align*}
			Concerning (2), we note that $(z-P)^{-1}\jap{r}^s\{\frac{(n-1)^2-2(n-1)}{4r^2}J + \frac{n-1}{2r} \tilde J\}\jap{r}^{-s}$ is bounded as $\jap{r}^s\{\frac{(n-1)^2-2(n-1)}{4r^2}J + \frac{n-1}{2r} \tilde J\}$ is bounded.
			
			Also, $(z-P)^{-1}\jap{r}^s\pa_r\tilde J\jap{r}^s$ is bounded since $\pa_r$ is $P$-bounded and $\jap{r}^{2s}\tilde J$ is bounded for any $s$.
			
			For $\f \in C^\infty_c(M_f)$, we see
			\begin{align*}
				& \tilde J \pa_r \f\\
				& =\pa_r \tilde J \f +\{-r^{-(n-1)/2}(\pa^2_r j(r)) + \frac{n-1}{2r}r^{-(n-1)/2}(\pa_r j(r))\}\f.
			\end{align*}
			Thus we obtain $[\tilde J,\pa_r] =-r^{-(n-1)/2}{\pa^2_r j(r)} + \frac{n-1}{2r}r^{-(n-1)/2}{\pa_r j(r)}$ and so $(z-P)^{-1} \jap{r}^s[\tilde J,\pa_r]  \tilde J\jap{r}^s$ is bounded.
			To sum up, $(z-P)^{-1}\jap{r}^s\tilde J\pa_r\jap{r}^s= (z-P)^{-1} \jap{r}^s[\tilde J,\pa_r] \tilde J\jap{r}^s + (z-P)^{-1}\jap{r}^s\pa_r\tilde J\jap{r}^s$ is bounded. From these calculation and definition of $\tilde j$ and the assumption for $P$, we obtain (2).
			
			Concerning (3), first we calculate as following,
			\vspace{-2mm}
			\begin{align*}
				& \jap{r}^s\chi(P)(PJ-JP_f)\jap{r}^s\\
				& =  \chi(P) \jap{r}^s (PJ-JP_f)\jap{r}^s + [\jap{r}^s,\chi(P)](PJ-JP_f)\jap{r}^s\\
				& =  \chi(P)(P-i) (P-i)^{-1}\jap{r}^s (PJ-JP_f)\jap{r}^s + [\jap{r}^s,\chi(P)](PJ-JP_f)\jap{r}^s.
			\end{align*}
			First term is bounded from (2). Thus the problem is the second term.
			
			From the Helffer-Sj\"ostrand formula, we obtain
			\begin{align*}
				& [\jap{r}^s,\chi(P)](PJ-JP_f)\jap{r}^s\\
				& =(2i\pi)^{-1}\int_{\co}\partial_{\bar{z}}\tilde \chi(z)[\jap{r}^s,(z-P)^{-1}]\,dz\wedge d{\bar{z}}(PJ-JP_f)\jap{r}^s\\
				& =(2i\pi)^{-1}\int_{\co}\partial_{\bar{z}}\tilde \chi(z)(z-P)^{-1}[ \jap{r}^s , P ](z-P)^{-1}\,dz\wedge d{\bar{z}}(PJ-JP_f)\jap{r}^s\\
				& =(2i\pi)^{-1}\int_{\co}\partial_{\bar{z}}\tilde \chi(z)(z-P)^{-1}[ \jap{r}^s , P ](z-P)^{-1}(P-i)\,dz\wedge d{\bar{z}}\\
				& \times (P-i)^{-1}(PJ-JP_f)\jap{r}^s
			\end{align*}
			
			Similarly to the proof of (2), we can see $(P-i)^{-1}(PJ-JP_f)\jap{r}^s$ is bounded. Concerning the integrability of above integral, we only have to prove the integrability in Im$z < 1$. That is because $\tilde \chi$ is the almost analytic extension and $(P-z)^{-1}$  has no critical point in Im$z > 1$.
			Since $\norm{(z-P)^{-1}(i-P)}=O({\mathrm{Im}z}^{-1})$ and $\norm{(z-P)^{-1}}=O({\mathrm{Im}z}^{-1})$ if $\mathrm{Im}z < 1$, the integrand is bounded uniformly in $z$, where we have used the fact $\bigabs{\partial_{\bar{z}}\tilde\chi(z)}=O({(\mathrm{Im}z)}^2)$ as it is almost analytic extension.
			Thus we obtain Proposition \ref{prop:2}.
		\end{proof}
		
	\section{Existence and completeness of $W_\pm$}
	We prepare some lemmas to prove Theorem \ref{thm:1}.
	
	\begin{lem}\label{lem:3}
		Let $\chi\in C^\infty_c(\re)$ then $\chi(P)J-J\chi(P_f)$ is a compact operator.
	\end{lem}
	\begin{proof}
		From the Helffer-Sj\"ostrand formula, we see
		\begin{align*}
		& \chi(P)J-J\chi(P_f)\\
		& =(2i\pi)^{-1}\int_{\co}\partial_{\bar{z}}\tilde \chi(z)\{(z-P)^{-1}J-J(z-P_f)^{-1}\}\,dz\wedge d{\bar{z}}\\
		& =-(2i\pi)^{-1}\int_{\co}\partial_{\bar{z}}\tilde \chi(z)(z-P)^{-1}(PJ-JP_f)(z-P_f)^{-1}\,dz\wedge d{\bar{z}}
		\end{align*}
		From Proposition \ref{prop:2}, it is sufficient to prove the following claim. 
		\begin{itemize}
			\item{Let $T$ be an operator from $\mathcal{H}_f$ to $\mathcal{H}$. We assume $(z-P)^{-1}\jap{r}^s T \jap{r}^s$ is bounded for some $1/2<s<1$ and any $z\in \co$. Then one can see that $(z-P)^{-1}T(z-P_f)^{-1}$ is a compact operator.
			}
		\end{itemize}
		We can calculate as following,
		\begin{align*}
		& (z-P)^{-1}T(z-P_f)^{-1}\\
		& =\jap{r}^{-s}(z-P)^{-1}\jap{r}^s T \jap{r}^s \times \jap{r}^{-s}(z-P_f)^{-1}\\
		& +\jap{r}^{-s}(z-P)^{-1}[P,\jap{r}^{s}](z-P)^{-1} T \jap{r}^s\jap{r}^{-s}(z-P_f)^{-1}.
		\end{align*}
		From the assumption for the claim and the locally compactness of $P_f$, first term is compact. Since $s<1$ implies $[P,\jap{r}^{s}]$ is $P$-bounded and the fact $(z-P)^{-1} T \jap{r}^s$ is bounded can be proved similarly to the proof of Proposition 7, second term is also compact and we conclude the proof of the claim.
		
		We apply this claim with $T=(PJ-JP_f)$ to obtain Lemma \ref{lem:3}.
	\end{proof}
	\begin{lem}\label{lem:4}
		Let $\chi\in C^\infty_c(\re)$ then $\jap{r}^s\chi(P)\jap{r}^{-s}$ is bounded operator on $\mathcal{H}$ for any $1/2<s<1$.
	\end{lem}
	\begin{proof}
		From the Helffer-Sj\"ostrand formula, we learn
		\begin{align*}
		& \jap{r}^s\chi(P)\jap{r}^{-s}\\
		& =(2i\pi)^{-1}\int_{\co}\partial_{\bar{z}}\tilde \chi(z)\jap{r}^s(z-P)^{-1}\jap{r}^{-s}\,dz\wedge d{\bar{z}}.
		\end{align*}
		Thus it is sufficient to prove the integrability of $\partial_{\bar{z}}\tilde \chi(z)\jap{r}^s(z-P)^{-1}\jap{r}^{-s}$.	We may assume that Im$z<1$ since the integrability on Im$z>1$ is obvious.
		
		 From the direct computations, we obtain
		\begin{align*}
			& \jap{r}^s(z-P)^{-1}\jap{r}^{-s}\\
			& = (z-P)^{-1} + (z-P)^{-1}[P,\jap{r}^s](z-P)^{-1}\jap{r}^{-s}\\
			& = (z-P)^{-1}\\
			&  + (z-P)^{-1}(i-P) \times (i-P)^{-1}[P,\jap{r}^s] \times (z-P)^{-1}\jap{r}^{-s}.\\
		\end{align*}
		From the fact that the order of $\norm{(z-P)^{-1}(i-P)}$ and $\norm{(z-P)^{-1}}$ is $(\mathrm{Im}z)^{-1}$, 
		and that $[P,\jap{r}^s]$ is $P$-bounded for $s\leq1$, $\partial_{\bar{z}}\tilde \chi(z)\jap{r}^s(z-P)^{-1}\jap{r}^{-s}$ is integrable on Im$z<1$ as $\tilde \chi(z)$ is the almost analytic extension of $\chi \in C^\infty_c(\re)$.
	\end{proof}
	\begin{proof}[Proof of Theorem \ref{thm:1}]
		We only have to prove the existence of the wave operator for $\f \in \mathcal{H}_f$ such that there exists 	$I \Subset \re\setminus (\s_{pp}(P) \bigcup \s_{pp}(P_f) \bigcup \mathrm{Cv}(\tilde V))$ which satisfies $ \f =\chi_I(P_f)\f$. That is because $\s_{pp}(P) \bigcup \s_{pp}(P_f) \bigcup \mathrm{Cv}(\tilde V)$ is discrete from assumption and Theorem 2, and hence such $\f$ are dense in $H_{ac}(P_f)$.
		
		For such $\f$, we obtain
		\begin{align*}
		& e^{itP} J e^{-itP_f}P_{ac}(P_f)\f\\
		& =e^{itP} \chi (P) J \chi (P_f)  e^{-itP_f}P_{ac}(P_f)\f + e^{itP}(J \chi (P_f)-\chi(P)J)e^{-itP_f}P_{ac}(P_f) \f\\
		\end{align*}
		From Lemma \ref{lem:3}, $(J \chi (P_f)-\chi(P)J)$ is compact. Thus we obtain 
		\begin{align*}
		\lim_{t\to\pm \infty}\norm{e^{itP}(J \chi (P_f)-\chi(P)J)e^{-itP_f}P_{ac}(P_f)\f}=0. 
		\end{align*}
		
		Then it is sufficient to prove $\lim_{t\to\pm \infty}e^{itP} \chi (P) J \chi (P_f)  e^{-itP_f}P_{ac}(P_f)\f$ exist.
		
		By differentiating $e^{itP} \chi (P) J \chi (P_f)  e^{-itP_f}P_{ac}(P_f)$, we obtain the following,
		\begin{align*}
		& \dfrac{d}{dt}\{e^{itP} \chi (P) J \chi (P_f)  e^{-itP_f}P_{ac}(P_f)\f\}\\
		& =ie^{itP} \chi (P) (PJ-JP_f) \chi (P_f)  e^{-itP_f}P_{ac}(P_f)\f\\
		& =ie^{itP} A^*_1 A_2 e^{-itP_f}P_{ac}(P_f)\f,
		\end{align*}
		where 
		\begin{align*}
		& A_1= \jap{r}^s(P_fJ^*-J^*P)\chi(P)\\
		& A_2=\jap{r}^{-s}\chi(P_f).
		\end{align*}
		
		From direct computations, we obtain
		\begin{align*}
		& A_1(P-\l \pm i \e)^{-1}A^*_1\\
		& = \jap{r}^s(P_fJ^*-J^*P)\chi(P) (P-\l \pm i \e)^{-1} \chi (P) (PJ-JP_f)\jap{r}^s\\
		& = \jap{r}^s(P_fJ^*-J^*P)\chi(P)\jap{r}^s\\
		& \times \jap{r}^{-s} (P-\l \pm i \e)^{-1}\jap{r}^{-s}\\
		& \times \jap{r}^s \chi (P) (PJ-JP_f)\jap{r}^s.
		\end{align*}
		Proposition \ref{prop:2} implies first and third part are bounded. Theorem \ref{thm:2} yields second part is bounded uniformly for $\e >0$. Thus $A_1$ is relatively $P$-smooth from Kato's characterization of relatively smoothness, especially
		\[
		\frac1{2\pi} \sup_{\norm{u}=1,u\in\dom(P)}\int^{\infty}_{-\infty}	\norm{A_1e^{-itP}u}^2 \mathrm{d}t < \infty.
		\]
		The fact that $A_2$ is $P_f$-smooth is proved directly from Theorem \ref{thm:2}.
		
		Let $\phi\in L^2(M)$, then we obtain 
		\begin{align*}
		& \bigabs{\jap{\phi,e^{itP} \chi (P) J \chi (P_f)  e^{-itP_f}P_{ac}(P_f)\f} }\\
		& =\bigabs{  -i\int^t_0\jap{\phi,e^{itP} \chi (P) (PJ-JP_f) \chi (P_f)  e^{-itP_f}P_{ac}(P_f)\f} dt\\
		& + \jap{\phi,\chi (P) J \chi (P_f) P_{ac}(P_f)\f} }\\
		& \leq \int^\infty_{-\infty}\bigabs{ \jap{\phi,e^{itP} \chi (P) (PJ-JP_f) \chi (P_f)  e^{-itP_f}P_{ac}(P_f)\f}} dt \\
		& +\bigabs{\jap{\phi,\chi (P) J \chi (P_f) P_{ac}(P_f)\f}}\\
		& = \int^\infty_{-\infty}\bigabs{ \jap{A_1e^{-itP}\phi, A_2 e^{-itP_f}P_{ac}(P_f)\f}} dt +\bigabs{\jap{\phi,\chi (P) J \chi (P_f) P_{ac}(P_f)\f}}\\
		& \leq  \int^\infty_{-\infty}\norm{A_1e^{-itP}\phi}\norm{ A_2 e^{-itP_f}P_{ac}(P_f)\f}dt +\bigabs{\jap{\phi,\chi (P) J \chi (P_f) P_{ac}(P_f)\f}}\\
		& \leq  (\int^\infty_{-\infty}\norm{A_1e^{-itP}\phi}^2dt )^{\frac 12}(\int^\infty_{-\infty}\norm{ A_2 e^{-itP_f}P_{ac}(P_f)\f}^2dt)^{\frac 12}\\
		& +\bigabs{\jap{\phi,\chi (P) J \chi (P_f) P_{ac}(P_f)\f}}.
		\end{align*}
		Thus $\lim_{t \to \pm \infty}e^{itP} \chi (P) J \chi (P_f)  e^{-itP_f}P_{ac}(P_f)\f$ exist and hence $W_\pm$ exist.
		
		We can prove the existence of $\tilde W_\pm=\slim_{t\to\pm \infty} e^{itP_f} J^* e^{-itP}P_{ac}(P)$ as follows:
		\begin{itemize}
			\item{Concerning $J^*$, counterpart of Proposition \ref{prop:2} is proved similarly by exchanging $P$ by $P_f$ and $J$ replaced with $J^*$. }
			\item{Then Lemma \ref{lem:3} is also proved similarly.}
			\item{Lemma \ref{lem:4} is directly proved by $P$ replaced by $P_f$.}
			\item{It is sufficient to prove $\lim_{t\to\pm \infty}e^{itP_f} \chi (P_f) J^* \chi (P)  e^{-itP}P_{ac}(P)\f$ exist to prove $\tilde W$ exist,
				 where $\f\in\mathcal{H_f}$ is such that $ \f =\chi_I(P_f)\f$ with $I \Subset \re\setminus (\s_{pp}(P) \bigcup \s_{pp}(P_f) \bigcup \mathrm{Cv}(\tilde V))$. }
			\item{Relatively smoothness argument in the proof of Theorem \ref{thm:1} can be recovered by replacing $A_1$ and $A_2$ by  $A^*_2$ and $A^*_1$. }
		\end{itemize}
	\end{proof}


\begin{thebibliography}{99}
			\bibitem{ach} S. Agmon, J. Cruz and I. Herbst, Generalized Fourier transform for Schr\"odinger operators with potentials of order zero, \textit{J. Funct. Anal.} \textbf{167} (1999), 345-369.
			\bibitem{hmv} A. Hassell, R. B. Melrose and A. Vasy, Spectral and scattering theory for	symbolic potentials of order zero, \textit {Advances in Math.} \textbf{181} (2004), 1-87.
			\bibitem{hmv2} A. Hassell, R. B. Melrose and A. Vasy, Micro local propagation near radial points and scattering for symbolic potentials of order zero, \textit {Analysis and PDE} \textbf{1} (2008), No.2, 127-196.
			\bibitem{h} I. Herbst, Spectral and scattering theory for Schr\"odinger operators with potentials independent	of $|x|$, \textit {Amer. J. Math.} \textbf{113} (1991), 509-565.
			\bibitem{hs} I. Herbst and E. Skibsted, Quantum scattering for potentials independent	of $|x|$: Asymptotic completeness for high and low energies, \textit{Comm. PDE} \textbf{29} (2004), No.3-4, 547-610.
			\bibitem{hs2} I. Herbst and E. Skibsted, Quantum scattering for homogeneous of degree zero potentials:
			Absence of channels at local maxima and saddle points, \textit{Tech. report,} Center for Mathematical Physics and Stochastics, (1999)
			\bibitem{in} K. Ito and S. Nakamura, Time-dependent scattering theory for Schr\"odinger operators on scattering manifolds, \textit{J. London Math. Soc.} \textbf{81} (2010), 774-792.
			\bibitem{in2} K. Ito and S. Nakamura,Microlocal properties of scattering matrices for Schr\"odinger equations on scattering, \textit{Anal. PDE} \textbf{6} (2013), No. 2, 257-286.
			\bibitem{mel} R. Melrose, Spectral and scattering theory for the Laplacian 
			on asymptotically Euclidean spaces, \textit {Spectral and Scattering Theory} (M. Ikawa, ed.)(1994), 85-130, Marcel Decker.
			\bibitem{ys} Y. Sait\={o}, Schr\"odinger operators with a non spherical radiation condition, \textit{Pacific J. Math.} \textbf{126} (1987), No. 2, 331-359.
			\bibitem{ya} D.R. Yafaev, Mathematical Scattering Theory: Analytic Theory, \textit{American Mathematical Society} (2009)
			\bibitem{zws}  M. Zworski,  Semiclassical analysis, Graduate Studies in Mathematics, \textit{American Mathematical Society} \textbf{138} (2012)
	\end{thebibliography}
\end{document}